%% file: root_arxiv.tex
\documentclass[letterpaper, 10 pt, conference]{ieeeconf}  

\IEEEoverridecommandlockouts                              

\usepackage{graphics} 
\usepackage{color}
\usepackage{amsmath} 
\usepackage{amssymb}  
\usepackage{amsfonts}
\usepackage{empheq}
\usepackage{comment}
\usepackage{tikz}
\usepackage{algpseudocode}
\usepackage{algorithm}
\usepackage{subcaption}

\usepackage{microtype}

\usepackage{hyperref}

\newtheorem{problem}{Problem}

\newtheorem{theorem}{Theorem}
\newtheorem{proposition}{Proposition}
\newtheorem{corollary}{Corollary}

\newtheorem{definition}{Definition}
\newtheorem{constraint}{Constraint}
\newtheorem{example}{Example}
\newtheorem{remark}{Remark}

\newcommand{\R}{\mathbb{R}}

\newcommand{\REV}[1]{#1}

\renewcommand{\AA}[1]{{\color{cyan}AA: #1}}

\title{\LARGE \bf
A Low Rank Approach to Minimize Sensor-to-Actuator Communication in Finite Horizon Output Feedback
}

\author{Antoine Aspeel$^{1}$, Jakob Nylof$^{2}$, Jing Shuang (Lisa) Li$^{1}$ and Necmiye Ozay$^{1}$
\thanks{*This work is funded by the ONR grant N00014-21-1-2431 (CLEVR-AI) and NSF Grant CNS-1837680.}
\thanks{$^{1}$A. A., J. S. L., and N. O. are with the Electrical Engineering and Computer Science Department, Univ. of Michigan, Ann Arbor, MI {\tt\small antoinas,jslisali,necmiye@umich.edu}.
}%
\thanks{$^{2}$J. N. is with KTH Royal Institute of Technology, Stockholm, Sweden, and he was a visiting scholar at the Univ. of
Michigan {\tt\small jnylof@kth.se}.
}%
}

\begin{document}

\maketitle
\thispagestyle{empty}
\pagestyle{empty}

\begin{abstract}
Many modern controllers are composed of different components that communicate in real-time over some network with limited resources. In this work, we are interested in designing a controller that can be implemented with a minimum number of sensor-to-actuator messages, while satisfying safety constraints over a finite horizon. For finite horizon problems, a linear time-varying controller with memory can be represented as a block-lower-triangular matrix. We show that the rank of this matrix exactly captures the minimum number of messages needed to be sent from the sensors to actuators to implement such a controller.
Moreover, we introduce a novel matrix factorization called \textit{causal factorization} that gives the required implementation. Finally, we show that the rank of the controller is the same as the rank of the Youla parameter, enabling the Youla parametrization (or analogous parametrizations) to be used to design the controller, which reduces the overall design problem into a rank minimization one over a convex set. Finally, convex relaxations for rank are used to demonstrate that our approach leads to 20-50\% less messages on a simulation than a benchmark method.

\end{abstract}

\section{Introduction} 

\REV{In a growing number of real-world applications (e.g., wireless sensor networks), controllers are implemented using distributed components (i.e., sensors and actuators) which must coordinate via limited resources. This work particularly concerns scenarios where sensors and actuators are not collocated, as in smart building heating systems, drone control in a motion capture arena, or on factory floors. In these cases, sensors and actuators may send messages to one another over a communication network with channel bandwidth constraints \cite{cervin2003does}. One resource-minimizing approach is to transmit fewer messages along the communication network to reduce the burden on the network. In this work, we consider the problem of designing a controller that can be implemented with a minimum number of sensor-to-actuator messages, while satisfying safety constraints. }

\subsubsection{Related \REV{works}}\label{sec:related_works}

The use of a communication network between distributed components to control a dynamic system induces numerous technical challenges (e.g., packet losses, communication delays and quantization due to limited bandwidth). The field of \emph{networked control systems} has emerged to tackle these challenges. In this field, a common objective is to minimize the use of sensors and actuators \cite{aspeel2021optimal}, \REV{\cite{balaghi20192, long2017stochastic}}, which gave rise to \emph{self-} and \emph{event-triggered controllers} \cite{heemels2012introduction}. Self-triggered controllers decide when to take the next action while computing the current action; event-triggered controllers trigger a control action only when the state (or state estimate) satisfies a certain condition. \REV{A related work \cite{BRAKSMAYER20172633} proposed a procedure to approximate a pre-designed linear stabilizing controller with a sensor- and an actuator-side controller that exchange messages only at given time instances to allow resource-aware implementations.}  However, the question of minimizing the number of messages is not as well characterized in the literature. In the context of static state-feedback quadratic control, a related recent work \cite{cho2023lowrank} proposed \REV{the use of a low-rank} gain matrix to reduce the energy used in broadcast communication between agents over a wireless network. However, our work differs substantially from \cite{cho2023lowrank} as we consider the time-varying, output feedback, and safety-constrained setting. In particular, the time-varying formulation allows us to optimize ``when" and ``what" to communicate, where we show that rank indeed is the correct metric to minimize sensor-to-actuator communication. 

In the related field of \emph{distributed control}, several physically interconnected (i.e., dynamically coupled) subsystems equipped with local controllers communicate with each \REV{other} via a communication network \cite{antonelli2013interconnected}. In this context, we might want to constrain which subsystems communicate with which others. This translates into sparsity constraints on the controller, which generally leads to NP-Hard problems \cite{wang2018convex}. While sparsity constraints are simple linear constraints on the controller gains, most problems are not convex in the controller gains directly, even in the centralized setting. This has led to the study of several \REV{parametrizations}, like \emph{Youla or Q-parametrization} \cite{skaf2010design} or \emph{\REV{system-level} synthesis} (SLS) \cite{anderson2019system}, where the control design problem is rendered convex after a suitable nonlinear change of variables. Then, the complexity of the distributed control design problem depends on whether the 
\REV{constraints} remain convex after the change of variables. Indeed, any convex constraint on the controller gain translates into a convex constraint over the Youla parameter if and only if it is \emph{quadratically invariant} (QI) \cite{rotkowitz2005characterization, lessard2011quadratic}. However, not all sparsity constraints are QI \cite[Section 3.5]{anderson2019system}, limiting the use of Youla parametrization for distributed control. To circumvent these limitations, SLS was developed \cite{anderson2019system, alonso2022distributed} --- a key 
feature of this approach is that it admits natural controller implementations in terms of the SLS parameters.

An alternative to having a fixed sparsity pattern constraining the communication architecture is to co-optimize the placement of sensors, actuators, and communication links by promoting sparsity \cite{lin2013design, matni2016regularization}. Convex relaxations of the sparsity maximization problem are proposed to minimize the number of sensors, actuators, and communication links. Though very closely related to the problem we consider, this framework does not explicitly consider the problem of minimizing messages --- in the aforementioned works, if a communication link is placed, it is assumed to be free of limitations (e.g., bandwidth). Also related are the works on sensor/actuator scheduling to maintain observability/controllability while promoting sparsity \cite{siami2020deterministic}.

\subsubsection{Contributions}

We consider a linear dynamical system subject to safety constraints over a finite horizon and address the problem of synthesizing an output feedback control law with memory while \emph{minimizing the number of messages sent from the sensors to the actuators}. In this \REV{finite-horizon} setting, linear memoryful controllers can be represented by a block-lower-triangular matrix $\mathbf{K}$.

\REV{Our three main contributions are as follows.} First, we prove that minimizing the number of messages can be done by minimizing the rank of $\mathbf{K}$. Second, we introduce the \emph{\REV{causal} factorization} of a matrix, and provide an algorithm to compute it. We prove that a minimum-message controller implementation can be obtained by computing the causal factorization of $\mathbf{K}$. Third, we prove that the rank of the controller is always equal to the rank of the Youla parameter, enabling the use of both Youla and SLS parametrization to design the controller. Finally, we use convex heuristics for rank minimization to demonstrate numerically that our method leads to fewer messages than \REV{sparsity-based} approaches.

\subsubsection{Notation} $\operatorname{rank} A$ and $\operatorname{Im} A$ denote the rank and image of the matrix $A$, respectively. $A_{i, :}$ is the $i$-th row of $A$ and $A_{l:k, :}$ the submatrix of $A$ formed by the rows $\{i\mid l\leq i\leq k\}$ (similarly for columns). $\operatorname{blkdiag}(A_1,\dots,A_n)$ denotes a block-diagonal matrix with diagonal blocks $A_1,\dots,A_n$. A matrix $A\in \R^{Tm\times Tn}$ is $(m,n)$-block-lower-triangular if $A_{tm+1:(t+1)m, \tau n+1:(\tau+1)n} = 0$ for all $0\leq t < \tau \leq T-1$. $\mathcal{X}\times\mathcal{Y}$ is the Cartesian product between sets $\mathcal{X}$ and $\mathcal{Y}$, \REV{and $\mathcal{X}^n$ is the $n$-th Cartesian power of the set $\mathcal{X}$}.

\section{Problem statement}
Consider a linear time-varying discrete time system
\begin{align}\label{eq:dynamics}
x_{t+1}=A_t x_t+B_t u_{t}+w_t,\ \ \ \ 
y_t=C_t x_{t}+v_t
\end{align}
where $x_t\in\R^{n_x}$, $u_t\in\R^{n_u}$, $w_t\in\R^{n_x}$, $y_t\in\R^{n_y}$, and $v_t\in\R^{n_y}$ are state, actuation, process noise, sensor measurement, and measurement noise, respectively. The finite horizon $t=0,\dots,T$ is considered.


We aim to minimize sensor-to-actuator messages while constraining the system trajectories for any disturbances in a bounded set. In particular, we consider polyhedra $\mathcal{X}_t$, $\mathcal{W}_t\subset\R^{n_x}$, $\mathcal{V}_t\subset\R^{n_y}$, $\mathcal{U}_t\subset\R^{n_u}$ and define the following constraint.
\begin{constraint}[Safety] \label{constraint:safety}
For all $w_t\in\mathcal{W}_t$, $v_t\in\mathcal{V}_t$ for $t=0,\dots,T-1$ and for all $x_0\in\mathcal{X}_0$, it holds that $u_t\in\mathcal{U}_t$ for $t=0,\dots,T-1$, and $x_t\in\mathcal{X}_t$ for $t=1,\dots,T$.
\end{constraint}

Before formalizing the problem of message minimization, we consider a motivating example.
\begin{example}[Minimizing messages]
\label{example:minimizing_messages}
Let \eqref{eq:dynamics} be a SISO system ($n_y = n_u = 1$), $T=4$ and let
\begin{align*}
\REV{\mathbf{y} = \begin{bmatrix}y_0 \\ y_1 \\ y_2 \\ y_3\end{bmatrix}, \
\mathbf{u} = \begin{bmatrix}u_0 \\ u_1 \\ u_2 \\ u_3\end{bmatrix}}, \
\mathbf{K}=\begin{bmatrix}
5 & & & \\
10 & 0 & & \\
0 & 3 & 4 & \\
15 & 6 & 8 & 0
\end{bmatrix}.
\end{align*}
To implement the controller $\mathbf{u=Ky}$, one can observe that $\mathbf{K}_{4,4}=0$, and conclude that only 3 sensor-to-actuator messages are needed ($y_3$ is not used). However, we show that this controller can be implemented with only 2 such messages. First, note that
$
\mathbf{K=DE}:=\REV{\begin{bmatrix}
1 & 0 \\
2 & 0 \\
0 & 1 \\
3 & 2
\end{bmatrix}
\begin{bmatrix}
5 & 0 & 0 & 0 \\
0 & 3 & 4 & 0 \\
\end{bmatrix}}.
$
This shows that the measurements can be encoded by the matrix $\mathbf{E}$ before the transmission and then be decoded using matrix $\mathbf{D}$ (see Fig. \ref{fig:block_diagram}). On the sensor side, the message $m_1=\mathbf{E}_{1,:}\mathbf{y}=5y_0$ can be sent at time $t_1=0$, and the message $m_2=\mathbf{E}_{2,:}\mathbf{y}=3y_1+4y_2$ can be sent at time $t_2=2$. On the actuator side, the inputs can be recovered as $u_0=1m_1$, $u_1=2m_1$, $u_2=1m_2$, and $u_3=3m_1+2m_2$. Crucially, the causality is respected: each message is transmitted (i) after the encoded measurements have been measured, but (ii) before being used by the actuator. This allows sending only two messages.
$\hfill\blacktriangle$
\end{example}
\vspace{-.2cm}
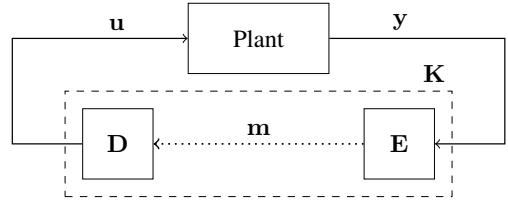
\begin{figure}[h]
\centering
\resizebox{.8\columnwidth}{!}{
\input{figures/block_diagram}
}
\caption{\small Block representation of the encoder-decoder structure of the controller.}
\vspace{-0.3cm}
\label{fig:block_diagram}
\end{figure}

Inspired by the previous example, we define the controller with the following encoder-decoder structure: for $0\leq t_1 \leq t_2 \leq \dots \leq t_r \leq T$, let
\begin{align}\label{eq:controller}
m_k = \sum_{\tau\leq t_k} e^\top_{(k,\tau)}y_\tau\text{ and }
u_t = \sum_{k \text{ s.t. }t_k\leq t} d_{(t,k)}m_k
\end{align}
with $e_{(k, \tau)}\in \R^{n_y}$, $d_{(t, k)}\in \R^{n_u}$, and each $m_k\in\R$ is a message sent from the sensor to the actuator at time $t_k$ (each message is a real number). For the SISO system in Example~\ref{example:minimizing_messages}, $e_{(k, \tau)}^\top=\mathbf{E}_{k, \tau+1}\in\R$ and $d_{(t, k)}=\mathbf{D}_{t+1, k}\in\R$.

\begin{remark}
For simplicity, we consider the messages $m_k$ sent from sensors to actuators to be real numbers, which can in theory require infinite bandwidth. In practice, the messages exchanged must be quantized, e.g., implemented in \REV{fixed-point} arithmetic. This can be handled by including rounding errors \REV{in} the noise terms in \eqref{eq:dynamics}. A more indepth study of this issue to obtain \REV{information-theoretic} bounds as in \cite{tatikonda2004control}, \REV{\cite{nair2007feedback}, and references therein,} is left for future work.$\hfill\blacklozenge$ 
\end{remark}

The problem of minimizing the number $r$ of sensor-to-actuator messages can be written as follows.

\begin{problem} \label{prob:encoderDecoder}
Find the minimal $r$ such that there exist $\{t_k\}_{k=1}^r$, $\{e_{(k,\tau)}\}_{k=1,\dots,r}^{\tau=0,\dots,t_k}$ and $\{d_{(t,k)}\}_{k=1,\dots,r}^{t=0,\dots,t_k}$ satisfying \eqref{eq:dynamics}, \eqref{eq:controller} and Constraint \ref{constraint:safety}.
\end{problem}

\begin{remark}
Enforcing $d_{(t,l)}=d_{(t_k,l)}$ for all $l\leq k$ is equivalent to enforcing zero order hold constraints $u_t=u_{t_k}$ for all $t$ such that $t_k\leq t <t_{k+1}$. As such, optimizing over the $d_{(t,l)}$ can be interpreted as optimizing the (linear) holding mechanism, i.e., what control input should be used when there is no new measurement. 
$\hfill\blacklozenge$
\end{remark}

\section{Method} \label{sec:method}

\REV{The controller structure \eqref{eq:controller} can be rewritten compactly \REV{as} $\mathbf{m}=\mathbf{Ey}$, $\mathbf{u}=\mathbf{Dm}$ by defining} $\mathbf{u}\coloneqq\begin{bmatrix} u_0^\top & \dots & u_T^\top\end{bmatrix}^\top$, $\mathbf{y}\coloneqq\begin{bmatrix} y_0^\top & \dots & y_T^\top\end{bmatrix}^\top$, $\mathbf{m}\coloneqq\begin{bmatrix} m_1^\top & \dots & m_r^\top \end{bmatrix}^\top$ and
\begin{small}
\begin{align*}
    &\mathbf{D} \coloneqq\begin{bmatrix}
    d_{(0,1)} & \cdots & d_{(0,r)} \\
    \vdots & & \vdots \\
    d_{(T,1)} & \cdots & d_{(T,r)}
    \end{bmatrix},\ 
    \mathbf{E} \coloneqq\begin{bmatrix}
    e^\top_{(1,0)} & \cdots & e^\top_{(1,T)} \\
    \vdots & & \vdots \\
    e^\top_{(r,0)} & \cdots & e^\top_{(r,T)}
    \end{bmatrix},
\end{align*}
\end{small}
with
\begin{align}\label{eq:causal_ed}
e_{(k,\tau)}=0 \text{ when } \tau>t_k, \text{ and } d_{(t,k)}=0 \text{ when } \REV{t_k>t}.
\end{align}
Matrices $\mathbf{E}$ and $\mathbf{D}$ can be interpreted as encoders and decoders, respectively. Note that \REV{because of \eqref{eq:causal_ed}} they satisfy the following causality constraints.
\begin{constraint}[Causal encoder/decoder]\label{constraint:causal_encoding}
$\mathbf{E}_{k,\tau n_y+j}=\mathbf{D}_{tn_u+i,k}=0$ for all $k=1,\dots,r$, for all $\tau>t_k>t$, for all $i=1,\dots,n_u$, and for all $j=1,\dots,n_y$.
\end{constraint}

\REV{Constraint \ref{constraint:causal_encoding} (or equivalently, conditions \eqref{eq:causal_ed}) states that a message transmitted at time $t_k$ can neither (i) encode measurements received after $t_k$, nor (ii) be used before $t_k$.} Importantly, any matrices satisfying Constraint \ref{constraint:causal_encoding} define some $\{e_{(k,\tau)}\}_{k=1,\dots,r}^{\tau=0,\dots,t_k}$ and $\{d_{(t,k)}\}_{k=1,\dots,r}^{t=0,\dots,t_k}$ in \eqref{eq:controller}. Consequently, one can optimize over the set of matrices $\mathbf{D}$ and $\mathbf{E}$ satisfying Constraint \ref{constraint:causal_encoding} instead of optimizing over $e_{(k,\tau)}$ and $d_{(t,k)}$.

\subsection{Reformulation as a controller with memory}
The encoder-decoder structure \eqref{eq:controller} of the controller can be written as a linear time-varying output feedback controller with memory
\begin{align}\label{eq:controller:linear}
u_{t}=\sum_{\tau\leq t}K_{(t,\tau)}y_\tau,
\end{align}
with $K_{(t,\tau)}=\sum_{k\text{ s.t. }\tau\leq t_k\leq t} d_{(t,k)}e^\top_{(k,\tau)}$. We use the following notation
\begin{small}
\begin{align} \label{eq:K:block}
    &\mathbf{K}\coloneqq\begin{bmatrix}
    K_{(0,0)} & & & \\
    K_{(1,0)} & K_{(1,1)} & & \\
    \vdots & \ddots & \ddots & \\
    K_{(T,0)} & \hdots & K_{(T,T-1)} & K_{(T,T)}
    \end{bmatrix}.
\end{align}
\end{small}

The controller structure can be written compactly as $\mathbf{u=Ky}$. It follows that $\mathbf{K=DE}$ is $(n_u,n_y)$-block-lower-triangular.

We aim to solve Problem \ref{prob:encoderDecoder} by optimizing over $\mathbf{K}$ instead of $(\mathbf{D},\mathbf{E})$ to avoid the bilinear term $\mathbf{DE}$. To do so, we need to know in which cases and how matrices $(\mathbf{D,E})$ can be recovered from $\mathbf{K}$. This is captured by the notion of \emph{causal factorization}.

\begin{definition}[Causal factorization]
Let $\mathbf{K}\in\R^{(T+1)n_u\times (T+1)n_y}$ be a $(n_u,n_y)$-block-lower-triangular matrix. A pair of matrices $(\mathbf{D,E})\in\R^{(T+1)n_u\times r}\times \R^{r \times (T+1)n_y}$
is a \emph{causal factorization of $\mathbf{K}$ with band $r$} if $\mathbf{K} = \mathbf{DE}$ and there exist integers $0\leq t_1\leq t_2\leq \dots \leq t_r \leq T$ such that Constraint \ref{constraint:causal_encoding} (causality) holds.
\end{definition}

Note that because $\mathbf{K=DE}$, there is no causal factorization with band smaller than $\operatorname{rank}\mathbf{K}$. We now state our first main result.

\begin{theorem} \label{thm:causal_factorization}
Any $(n_u,n_y)$-block-lower-triangular matrix $\mathbf{K}\in\R^{(T+1)n_u\times (T+1)n_y}$ admits a causal factorization $(\mathbf{D,E})$ with band equal to $\operatorname{rank}\mathbf{K}$. In addition, Algorithm \ref{algo:causalFactorization} returns such a causal factorization.
\end{theorem}

Before proving Theorem \ref{thm:causal_factorization}, let us discuss the intuition of Algorithm \ref{algo:causalFactorization}. The matrix $\mathbf{E}$ is constructed as follows: a row of $\mathbf{K}$ is added to $\mathbf{E}$ only if it is linearly independent of the rows preceding it. In this way, when a line from $\mathbf{K}$ is not added to $\mathbf{E}$, it can be reconstructed (linearly) from the preceding lines in $\mathbf{E}$ (this is how $\mathbf{D}$ is constructed).

\begin{proof}
It is enough to prove the second claim of Theorem~\ref{thm:causal_factorization}. To this end, first, note that for all $l=1,\dots,(T+1)n_u$:
\vspace{-.25cm}\begin{align*}
\mathbf{K}_{l,:}\in \operatorname{Im}\mathbf{K}_{1:l,:} = \operatorname{Im}\begin{bmatrix} 
\mathbf{K}_{c_1,:} \\
\mathbf{K}_{c_2,:} \\
\vdots \\
\mathbf{K}_{c_{r_l},:}
\end{bmatrix} = \operatorname{Im}\mathbf{E}_{1:r_l,:},
\end{align*}
where the first equality follows from the definition of $c_k$ (line \ref{algo:line:def_ck}), and the second equality follows from the definition of $\mathbf{E}$ (line \ref{algo:line:def_E}). Consequently, the system of linear equations in line \ref{algo:line:system} always has a solution, and $\mathbf{K=DE}$.

For $k=1,\dots, r$, let $t_k$ be such that $c_k=t_k n_u+i_k$ for some $i_k\in\{1,\dots,m\}$. \REV{Then,} $0\leq t_1\leq\dots\leq t_r\leq T$.

Let $k\in\{1,\dots,r\}$, $\tau>t_k$ and $j\in\{1,\dots,n_y\}$. We need to prove that $\mathbf{E}_{k,\tau n_y+j}=0$. By definition of $\mathbf{E}$ (line \ref{algo:line:def_E}), $\mathbf{E}_{k,\tau n_y+j}=\mathbf{K}_{c_k,\tau n_y+j}$. But by definition of $t_k$, $c_k\leq t_k n_u+n_u < \tau n_u+1$. It follows from the block-lower-triangularity of $\mathbf{K}$ that $\mathbf{E}_{k,\tau n_u+j}=0$.

Let $k\in\{1,\dots,r\}$, $t<t_k$ and $i\in\{1,\dots,n_u\}$. We need to prove that $\mathbf{D}_{tn_u+i,k}=0$. Note that $tn_u+i\leq (t+1)n_u \leq t_kn_u < t_k n_u+i_k=c_k$. Then, by definition of $c_k$ (line \ref{algo:line:def_ck}), $r_{tn_u+i}<r_{c_k}=k$. It follows from the definition of $\mathbf{D}$ (see line \ref{algo:line:def_D} with $l=tn_u+i$) that $\mathbf{D}_{tn_u+i,k}=0$.
\end{proof}

\begin{algorithm}
\small
\caption{Causal factorization} \label{algo:causalFactorization}
\begin{algorithmic}[1]
\Require $\mathbf{K}\in\R^{(T+1)n_u\times (T+1)n_y}$, $(n_u,n_y)$-block-lower-triangular.
\For{$l=1,\dots,(T+1)n_u$}
    \State $r_l\coloneqq \operatorname{rank}\mathbf{K}_{1:l,:}$ \label{algo:line:def_rl}
\EndFor
\State $r\coloneqq r_{(T+1)n_u}$ \Comment{$r=\operatorname{rank}\mathbf{K}$}
\For{$k=1,\dots,r$} \Comment{Compute $\mathbf{E}$}
    \State $c_k\coloneqq \min\{l\mid r_l=k\}$ \label{algo:line:def_ck}
    \State $\mathbf{E}_{k,:}=\mathbf{K}_{c_k,:}$ \label{algo:line:def_E}
\EndFor
\For{$l=1,\dots,(T+1)n_u$} \Comment{Compute $\mathbf{D}$}
    \State Find $\mathbf{D}_{l,1:r_l}$ such that $\mathbf{K}_{l,:}=\mathbf{D}_{l,1:r_l}\mathbf{E}_{1:r_l,:}$ \label{algo:line:system}
    \State $\mathbf{D}_{l,r_l+1:r}\coloneqq0$ \label{algo:line:def_D}
\EndFor
\State \Return $(\mathbf{D,E})$
\end{algorithmic}
\end{algorithm}

Causal factorization is not unique. Indeed, for any invertible diagonal matrix $\Lambda$, the pairs $(\mathbf{D,E})$ and $(\mathbf{D}\Lambda^{-1},\Lambda\mathbf{E})$ have the same band and factorize the same matrix. The factorization in Example \ref{example:minimizing_messages} is the one computed using Algorithm \ref{algo:causalFactorization}.


It follows from Theorem \ref{thm:causal_factorization} that Problem \ref{prob:encoderDecoder} can be formulated as a rank minimization problem over $\mathbf{K}$.
\begin{corollary}
\label{thm:min_rank_K}
Optimal $r$, $\{t_k\}_{k=1}^r$, $\{e_{(k,\tau)}\}_{k=1,\dots,r}^{\tau=0,\dots,t_k}$ and $\{d_{(t,k)}\}_{k=1,\dots,r}^{t=0,\dots,t_k}$ for Problem \ref{prob:encoderDecoder} can be obtained by finding an optimal $\mathbf{K}^*$ for
\begin{align} \label{eq:min_rank_K}
\min_\mathbf{K}~\operatorname{rank}~\mathbf{K} \text{ s.t. \eqref{eq:dynamics}, \eqref{eq:controller:linear}, \eqref{eq:K:block}, Constraint \ref{constraint:safety}},
\end{align}
and computing a causal factorization of $\mathbf{K}^*$ with band equal to $\operatorname{rank}\mathbf{K}^*$.
\end{corollary}
\begin{proof}
Let $r\in\mathbb{N}$ and $\{t_k\}_{k=1}^r$. As already noted in Section \ref{sec:method}, there exist vectors $d_{(t,k)}$ and $e_{(k,\tau)}$ satisfying \eqref{eq:controller} if and only if there exist matrices $(\mathbf{D},\mathbf{E}) \in\mathbb{R}^{n\times r} \times \mathbb{R}^{r\times m}$ satisfying $\mathbf{u=DEy}$ and Constraint \ref{constraint:causal_encoding}. Such matrices exist if and only if there exist a block-lower-triangular $\mathbf{K}$ such that $\mathbf{u=Ky}$ that admits a causal factorization with band $r$ (this follows from the definition of causal factorization). Because the band of a causal factorization can not be smaller than the rank, and by Theorem \ref{thm:causal_factorization}, $\mathbf{K}$ admits a causal factorization with band $r$ if and only if $\operatorname{rank}\mathbf{K}\leq r$. Consequently, minimizing $r$ is equivalent to minimizing $\operatorname{rank}\mathbf{K}$. Finally, note that $\mathbf{u=Ky}$ and the block-lower-triangularity of $\mathbf{K}$ are equivalent to \eqref{eq:controller:linear} and \eqref{eq:K:block}. This concludes the proof.
\end{proof}

\begin{remark}[Relation to sparsity] \label{remark:sparsity:RFD}
The \textit{regularization for design} framework \cite{matni2016regularization} proposes a method to optimize (and minimize) placement of sensors, actuators, and communication links using sparsity-based approaches. This method was originally developed for static infinite-horizon controllers, but can be adapted directly to time-varying controllers over a finite horizon. In a time-varying controller, regularization for design can be used to optimize not only sensor/actuator placements but also sensor/actuator usage. For example, minimizing the use of actuators can be written $\min \sum_l\| \mathbf{K}_{l,:}\|_0$. Indeed, $\mathbf{K}_{tm+i,:}=0$ implies $u_i(t)=0$ and actuator $i$ is not used at time $t$. Though closely related to the problem at hand, this sparsity-based approach does not explicitly consider message minimization.

\end{remark}

\subsection{System level synthesis (SLS)} \label{sec:SLS}
To handle the safety Constraint \ref{constraint:safety} (which involves \textit{for all} quantifiers), we follow \cite{chen2019system} and use SLS to rewrite it linearly in terms of SLS parameters (note that \cite{chen2019system} considers state feedback). Then, we prove that rank minimization of $\mathbf{K}$ has a natural reformulation in terms of SLS (or Youla) parameters, allowing us to rewrite \eqref{eq:min_rank_K} as a rank minimization problem subject to linear constraints. 

We make the following definitions: $\mathbf{x}\coloneqq\begin{bmatrix} x_0^\top & \dots & x_T^\top\end{bmatrix}^\top$, $\mathbf{w}\coloneqq\begin{bmatrix} x_0^\top & w_0^\top & \dots & w_{T-1}^\top\end{bmatrix}^\top$, $\mathbf{v}\coloneqq\begin{bmatrix} v_0^\top & \dots & v_T^\top\end{bmatrix}^\top$, the matrix $Z$ is the block-downshift operator (identity matrices on the first sub-diagonal and zeros elsewhere) and $\mathcal{A}\coloneqq \operatorname{blkdiag}(A_0,\dots,A_{T-1},0)$, 
 $\mathcal{B}\coloneqq \operatorname{blkdiag}(B_0,\dots,B_{T-1},0)$, $\mathcal{C}\coloneqq \operatorname{blkdiag}(C_0,\dots,C_{T})$.


We have that $\eqref{eq:dynamics}$ and $\eqref{eq:controller}$ can be written
$\mathbf{x} = Z\mathcal{A}\mathbf{x}+Z\mathcal{B}\mathbf{u} + \mathbf{w}$, $\mathbf{y} = \mathcal{C}\mathbf{x} + \mathbf{v}$, $\mathbf{u} = \mathbf{K}\mathbf{y}$
and equivalently
\begin{align}
\label{eq:system_response}
\begin{bmatrix} \mathbf{x}\\ \mathbf{u} \end{bmatrix} = \begin{bmatrix} \mathbf{\Phi}_{xx} & \mathbf{\Phi}_{xy} \\ \mathbf{\Phi}_{ux} & \mathbf{\Phi}_{uy} \end{bmatrix} \begin{bmatrix} \mathbf{w}\\ \mathbf{v} \end{bmatrix},
\end{align}
with $\mathbf{\Phi}_{xx}=(I-Z\mathcal{A}-Z\mathcal{B}\mathbf{K}\mathcal{C})^{-1}$, $\mathbf{\Phi}_{xy}=\mathbf{\Phi}_{xx}Z\mathcal{B}\mathbf{K}$, $\mathbf{\Phi}_{ux}=\mathbf{K}\mathcal{C}\mathbf{\Phi}_{xx}$ and $\mathbf{\Phi}_{uy}=\mathbf{K} + \mathbf{K}\mathcal{C}\mathbf{\Phi}_{xx}Z\mathcal{B}\mathbf{K}$.
The following proposition is the basis for finite horizon output feedback SLS and gives a condition under which all block-lower triangular controllers $\mathbf{K}$ can be parameterized by block-lower triangular system response $\{\mathbf{\Phi}_{xx},\mathbf{\Phi}_{xy},\mathbf{\Phi}_{ux},\mathbf{\Phi}_{uy}\}$ (and vice versa).
\begin{proposition}[Adapted from {\cite[Lemma 1]{hassaan2022system}}] \label{thm:slp}
Over the horizon $t=0,\dots,T$, the system dynamics \eqref{eq:dynamics} with controller \eqref{eq:controller}, the following are true:
\begin{enumerate}
\item the affine subspace defined by
\begin{subequations}\label{eq:SLP}
\begin{align}
\begin{bmatrix} I-Z\mathcal{A} & -Z\mathcal{B} \end{bmatrix} \begin{bmatrix} \mathbf{\Phi}_{xx} & \mathbf{\Phi}_{xy} \\ \mathbf{\Phi}_{ux} & \mathbf{\Phi}_{uy} \end{bmatrix} &= \begin{bmatrix} I & 0 \end{bmatrix} \\
\begin{bmatrix} \mathbf{\Phi}_{xx} & \mathbf{\Phi}_{xy} \\ \mathbf{\Phi}_{ux} & \mathbf{\Phi}_{uy} \end{bmatrix} \begin{bmatrix} I-Z\mathcal{A} \\ -\mathcal{C} \end{bmatrix} &= \begin{bmatrix} I \\ 0 \end{bmatrix}
\end{align}
\end{subequations}
parameterizes all possible system responses \eqref{eq:system_response}.
\item for any block-lower-triangular matrices $\{\mathbf{\Phi}_{xx},\mathbf{\Phi}_{xy},\mathbf{\Phi}_{ux},\mathbf{\Phi}_{uy}\}$ satisfying \eqref{eq:SLP}, the controller $\mathbf{K}=\mathbf{\Phi}_{uy}-\mathbf{\Phi}_{ux}\mathbf{\Phi}_{xx}^{-1}\mathbf{\Phi}_{xy}$ achieves the desired system response \eqref{eq:system_response}.
\end{enumerate}
\end{proposition}

The safety Constraint \ref{constraint:safety} can be handled linearly thanks to SLS. Indeed, it can be written
$\tilde{\mathbf{\Phi}}\mathcal{N} \subseteq \mathcal{S}$, where $\mathcal{N}\coloneqq\mathcal{X}_0\times\bigtimes_{t=0}^{T-1} \mathcal{W}_t\times\bigtimes_{t=0}^{T-1}\mathcal{V}_t$, $\mathcal{S}\coloneqq \bigtimes_{t=1}^T \mathcal{X}_t \times \bigtimes_{t=0}^{T-1}\mathcal{U}_t$, and
$$
\tilde{\mathbf{\Phi}}\coloneqq
\begin{bmatrix}
(\mathbf{\Phi}_{xx})_{n_x+1:(T+1)n_x,:} & (\mathbf{\Phi}_{xy})_{n_x+1:(T+1)n_x,:} \\
(\mathbf{\Phi}_{ux})_{1:Tn_u,:} & (\mathbf{\Phi}_{uy})_{1:Tn_u,:}
\end{bmatrix}.
$$
By Farkas' lemma, this is equivalent to the existence of a matrix $\Lambda$ such that
\begin{align} \label{eq:polytope_containment}
\Lambda \geq 0,\ \ \ 
\Lambda H_\mathcal{N} = H_\mathcal{S} \tilde{\mathbf{\Phi}},\ \ \ 
\Lambda h_\mathcal{N} \leq h_\mathcal{S},
\end{align}
where we used the H-representation $\mathcal{P}=\{p|H_\mathcal{P} p\leq h_\mathcal{P}\}$ for $\mathcal{P}\in\{\mathcal{N},\mathcal{S}\}$. Importantly, these constraints are linear in $\Lambda$ and the SLS parameters.\footnote{Polytope containment constraints can also be handled using Youla parametrization (see e.g., \cite[Lemmas 3 and 4]{aspeel2021optimal}).}

To use SLS to solve problem \eqref{eq:min_rank_K}, we also need to relate the rank of $\mathbf{K}$ to the SLS parameters. 
That is our second main contribution.
\begin{theorem} \label{thm:rank_youla}
For any block-lower-triangular matrices $\{\mathbf{\Phi}_{xx},\mathbf{\Phi}_{xy},\mathbf{\Phi}_{ux},\mathbf{\Phi}_{uy}\}$ satisfying \eqref{eq:SLP}, $\operatorname{rank}\mathbf{\Phi}_{uy}=\operatorname{rank}\mathbf{K}$, where $\mathbf{K}\coloneqq \mathbf{\Phi}_{uy}-\mathbf{\Phi}_{ux}\mathbf{\Phi}_{xx}^{-1}\mathbf{\Phi}_{xy}$.
\end{theorem}
\begin{proof}
Assuming \eqref{eq:SLP}, it follows from the second statement in Proposition \ref{thm:slp} that equation \eqref{eq:system_response} holds. Then, one can write $\mathbf{\Phi}_{uy}=(I+\mathbf{\Phi}_{ux}Z\mathcal{B})\mathbf{K}$. Because $\mathbf{\Phi}_{ux}$, $Z$ and $\mathcal{B}$ are block-lower-triangular, strictly block-lower-triangular and block diagonal, respectively, $\mathbf{\Phi}_{ux}Z\mathcal{B}$ is strictly block-lower diagonal and $(I+\mathbf{\Phi}_{ux}Z\mathcal{B})$ is invertible. Consequently, $\operatorname{rank}\mathbf{K}=\operatorname{rank}\mathbf{\Phi}_{uy}$.
\end{proof}

\begin{remark}
The matrix $\mathbf{\Phi}_{uy}$ is the Youla parameter (see \cite[Equation (11)]{skaf2010design}), so this theorem states that the controller and the Youla parameter have the same rank. We note that the constraint $\mathbf{K}\in\{\mathbf{K}\mid \operatorname{rank}\mathbf{K}\leq r\}$ is QI\footnote{A set $\Omega$ is QI with respect to a plant if $\mathbf{K}P_{22}\mathbf{K}\in\Omega$ for all $\mathbf{K}\in\Omega$, where the matrix $P_{22}$ depends on the plant.} for any plant, while $\mathbf{K}\in\{\mathbf{K}\mid \operatorname{rank}\mathbf{K}= r\}$ is not.

Theorem \ref{thm:rank_youla} suggests that the minimal number of sensor-to-actuator messages does not depend on the parametrization --- in particular, \REV{message minimization} can be done using the standard, Youla, or SLS parametrization. SLS may be preferable as it allows us to consider extra constraints that are not QI; in Section \ref{sec:multi_sensors_actuators}, we show that SLS can handle cases where (i) several sensors do not share their measurements and (ii) several actuators do not share the messages they receive (whereas Youla parametrization does not allow this). $\hfill\blacklozenge$


\end{remark}

It follows from Proposition \ref{thm:slp}, the discussion on polytope containment, and Theorem \ref{thm:rank_youla} that \eqref{eq:min_rank_K} can be rewritten as a rank minimization problem subject to linear constraints.
\begin{corollary}\label{thm:min_rank_youla}
An optimal solution $\mathbf{K}^*$ to Problem \eqref{eq:min_rank_K} is given by an optimal solution $\{\mathbf{\Phi}_{xx}^*, \mathbf{\Phi}_{xy}^*, \mathbf{\Phi}_{ux}^*, \mathbf{\Phi}_{uy}^*\}$ of
\vspace{-.25cm}\begin{align} \label{eq:min_rank_phi}
\begin{aligned}
    \min_{\mathbf{\Phi}_{xx}, \mathbf{\Phi}_{xy}, \mathbf{\Phi}_{ux}, \mathbf{\Phi}_{uy}, \Lambda} \quad \operatorname{rank}~\mathbf{\Phi}_{uy}& \\
    \text{s.t. \ \eqref{eq:SLP}, \eqref{eq:polytope_containment} and }\mathbf{\Phi}_{xx},\ \mathbf{\Phi}_{xy}, \mathbf{\Phi}_{ux},\ \mathbf{\Phi}_{uy}& \\
    \text{ are $(n_u,n_y)$-block-lower-triangular},&
\end{aligned}
\end{align}
with $\mathbf{K}^*=\mathbf{\Phi}_{uy}^*-\mathbf{\Phi}_{ux}^*(\mathbf{\Phi}_{xx}^*)^{-1}\mathbf{\Phi}_{xy}^*$.
\end{corollary}

\begin{remark}[Relation to sparsity continued] \label{remark:sparsity:QI}
Maximizing sparsity typically corresponds to minimizing the use of sensors/actuators. For example, minimizing the use of actuators reduces to maximizing the row sparsity of $\mathbf{K}$. Since row-sparsity is QI, $\min \sum_l\|\mathbf{K}_{l,:}\|_0 =\min \sum_l\| \mathbf{[\Phi}_{uy}]_{l,:}\|_0$.
$\hfill\blacklozenge$
\end{remark}

\subsection{Numerical considerations} \label{sec:numerical_considerations}
Problem \eqref{eq:min_rank_phi} is a rank minimization problem which is NP-hard. To solve it approximately, we use the \emph{reweighted nuclear norm heuristic} described in \cite[Section III]{mohan2010reweighted} with regularization parameter $\delta = 0.01$.

Once $\{\mathbf{\Phi}_{xx}^*, \mathbf{\Phi}_{xy}^*, \mathbf{\Phi}_{ux}^*, \mathbf{\Phi}_{uy}^*\}$ have been obtained, we compute the corresponding gain matrix $\mathbf{K}^*=\mathbf{\Phi}_{uy}^*-\mathbf{\Phi}_{ux}^*(\mathbf{\Phi}_{xx}^*)^{-1}\mathbf{\Phi}_{xy}^*$ (see Proposition \ref{thm:slp}). Then, its causal factorization $(\mathbf{D_{\epsilon},E_{\epsilon}})$ is computed using Algorithm \ref{algo:causalFactorization} (the ranks in line \ref{algo:line:def_rl} of the algorithm are computed with respect to a tolerance $\epsilon>0$). To handle the factorization error $\mathbf{K^*\approx K_\epsilon\coloneqq D_{\epsilon}E_{\epsilon}}$ due to $\epsilon\neq 0$, the feasibility of $\mathbf{K}_\epsilon$ is then checked by computing the corresponding $\{\mathbf{\Phi}_{xx}^\epsilon, \mathbf{\Phi}_{xy}^\epsilon, \mathbf{\Phi}_{ux}^\epsilon, \mathbf{\Phi}_{uy}^\epsilon\}$ (see equations following \eqref{eq:system_response}).

\begin{remark}[Relation to sparsity continued]\label{remark:sparsity:lasso}
Like minimizing rank, maximizing row sparsity is NP-hard. However, the 0-norm can be replaced by the 2-norm, leading to the (convex) actuator norm considered in \cite{matni2016regularization}. Alternatively, the row sparsity can be optimized using a reweighting heuristic \cite{candes2008enhancing}. From a computational point of view, minimizing the nuclear norm is a semidefinite program, while minimizing the sensor/actuator norm is a second order cone program (which can be done more efficiently). Another notable difference between sparsity and rank optimization is as follows: when maximizing sparsity, the controller can be re-optimized after the sparsity pattern has been fixed, since imposing a sparsity pattern is a linear constraint. On the contrary, the controller can not be re-optimized after the minimum rank has been found because rank constraints are not convex.$\hfill\blacklozenge$
\end{remark}

\subsection{Multiple sensors and actuators} \label{sec:multi_sensors_actuators}
Above, we assumed that there is only one sensor and one actuator (both of which make vector measurements/actuations). The case where several sensors do not share their measurements and several actuators do not share the messages they receive can be handled thanks to SLS (but can not be handled with Youla parametrization). For the sake of explanation, assume that each sensor $j$ makes a scalar measurement $y_j(t)\in\R$ and that each actuator $i$ actuates a scalar input $u_i(t)\in\R$. Using, for example, the controller implementation in \cite[Fig. 11(c)]{anderson2019system}, the following sparsity constraints can be used to prevent communication between sensors:
$[\mathcal{C}\mathbf{\Phi}_{xy}]_{tn+j_1,\tau n+j_2}=0$, for all $t,\tau$ and $j_1\neq j_2$. Then, the number of messages sent from sensor $j$ to actuator $i$ is given by $\operatorname{rank}\mathbf{\Phi}_{uy}^{(i,j)}$, where $[\mathbf{\Phi}_{uy}^{(i,j)}]_{t,\tau}\coloneqq [\mathbf{\Phi}_{uy}]_{tm+i,\tau n+j}$. Finally, minimizing the total number of messages is solving $\min \sum_i\sum_j \operatorname{rank}\mathbf{\Phi}_{uy}^{(i,j)}$. A causal factorization can be computed for each $\mathbf{\Phi}_{uy}^{(i,j)}$ independently.

\section{Numerical demonstrations}
To illustrate our method,\footnote{The code that generates the figures and implements our algorithm is available at: \url{https://github.com/aaspeel/lowRankControl}. \REV{The reported computation times are obtained} using a laptop with a Quad-Core Intel i7 CPU and 16 GB of RAM.} we consider a drone represented by the two-dimensional double integrator dynamics $\ddot{p}^x=u^x$, $\ddot{p}^y=u^y$, where $(p^x,p^y)$ represent the $(x,y)$-position of a drone subject to a force $(u^x,u^y)$. The state of the system is defined as $x=\begin{bmatrix} p^x & p^y & \dot{p}^x & \dot{p}^y \end{bmatrix}^\top$. The dynamics is exactly discretized with unit discretization step and $T=20$. A process noise $w_t\in[-0.05,0.05]^4$ is considered. The position $\begin{bmatrix} p^x & p^y \end{bmatrix}^\top$ is measured with some additive noise $v_t\in[-0.05,0.05]^2$. The initial state is $x_0\in[-8,-6]^2\times \{0\}^2$. The input constraints are $u_t\in[-2,2]$. The time-varying safety constraints over the state are $x_{10}\in[5,9]\times[-9,-5]\times[-2,2]^2$, $x_{20}\in[5,9]^2\times[-1,1]^2$, and $x_t\in[-10,10]^2\times[-2,2]^2$ at all other time steps.

\begin{figure}[!ht]
    \label{fig:simulations}
    \centering
        \begin{subfigure}[b]{\columnwidth}
        \includegraphics[width=\textwidth]{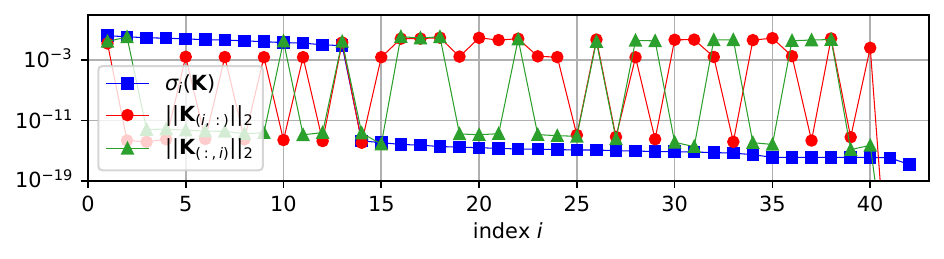}
        \caption{Ordered singular values of $\mathbf{K}$ for the nuclear norm case (squares). Column/row 2-norms of $\mathbf{K}$ for the sensor/actuator norm case (triangles/circles) after 8 iterations of the reweighting heuristics.}
        \label{fig:reweighting}
    \end{subfigure}
    \begin{subfigure}[b]{\columnwidth}
        \includegraphics[width=\textwidth]{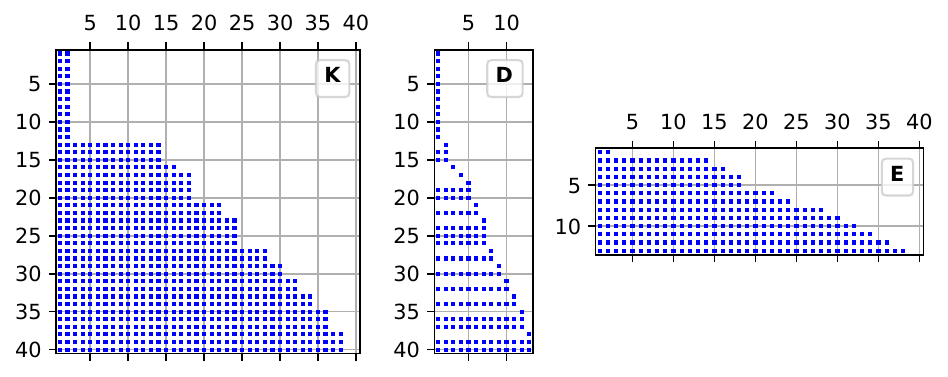}
        \caption{Sparsity of $\mathbf{K}$ and its causal factorization $\mathbf{(D,E)}$ for the nuclear norm case.}
        \label{fig:sparsity:nuclear}
    \end{subfigure}
    \begin{subfigure}[b]{\columnwidth}
        \centering
        \includegraphics[width=.91\textwidth]{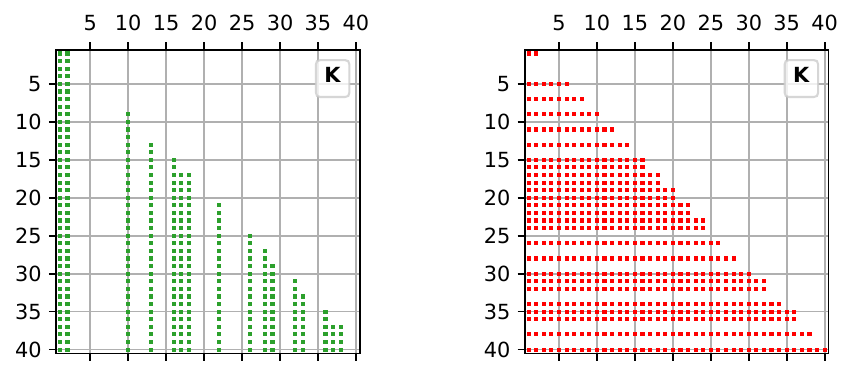}
        \caption{Sparsity of $\mathbf{K}$ for the sensor (left) and actuator norm case (right).}
        \label{fig:sparsity:sensor-actuator}
    \end{subfigure}
    \begin{subfigure}[b]{\columnwidth}
        \centering
       \includegraphics[width=.75\textwidth]{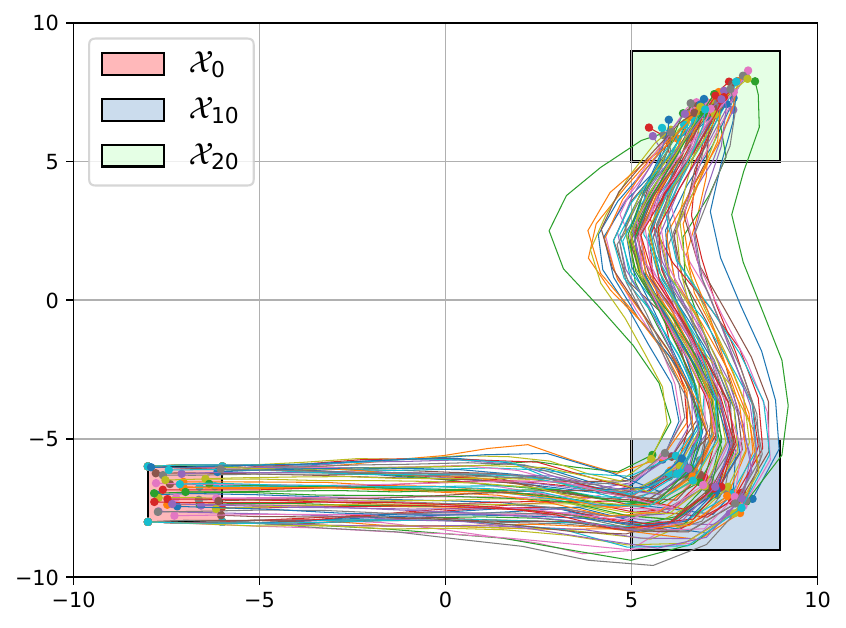}
        \caption{$(p^x_t,p^y_t)$ for the nuclear norm case (40 trajectories were generated from uniformly sampled noises and another 40 from noises over its vertices). The dots indicate the positions at times $t=0$, $10$ and $20$.}
        \label{fig:trajectories}
    \end{subfigure}
    \caption{\small Comparison of the solutions for the nuclear, sensor and actuator norm optimization problems.}
    \label{fig:enter-label}
    \vspace{-0.5cm}
\end{figure}

We compare our method to the sensor/actuator norms (see Remarks \ref{remark:sparsity:RFD}, \ref{remark:sparsity:QI}, \ref{remark:sparsity:lasso}) using the same reweighting heuristic \cite{candes2008enhancing}. Fig. \ref{fig:reweighting} shows the singular values (resp. the 2-norm of columns/rows) of $\mathbf{K}$. The large gaps indicate a small truncation error. Minimizing the rank took 205 seconds, while maximizing the column and row sparsity took 33 and 27 seconds, respectively. Fig. \ref{fig:sparsity:nuclear} presents the sparsity of the gain obtained from rank minimization and its causal factorization (computed in 0.01 seconds). Fig. \ref{fig:sparsity:sensor-actuator} shows the sparsity of the gain obtained by actuator and sensor norm minimization. Our method requires 13 messages (the band of the factorization), while sensor-/actuator-norms require 16 and 26 messages, respectively (the number of non-zero columns/rows). Finally, Fig. \ref{fig:trajectories} presents trajectories using the controller synthesized with our method. As expected, the trajectories satisfy the constraints.

\section{Conclusion}

\REV{We addressed the problem of synthesizing a controller that can be implemented with a minimum number of sensor-to-actuator messages. Our results show that for this objective, when considering linear time-varying controllers with memory, controller implementation benefits from having some computation on the sensor side (encoder) and some on the actuator side (decoder) as opposed to being fully collocated with the sensor or actuator. Our method relies on first minimizing the rank of the Youla parameter that we showed to be the same as the rank of the controller (with minimum rank giving the least number of message transmissions) and then computing the causal factorization of the controller. Future work will investigate infinite horizon counterparts of this method.}

\bibliographystyle{IEEEtran}
\bibliography{bibliography}

\end{document}

%% file: figures/block_diagram.tex
\begin{tikzpicture}[scale=1]
\centering
\def\boxHight{1cm}
\def\plantWidth{2cm}
\def\xlim{3.5cm}
\def\encCenter{2cm}
\def\encWidth{1cm}
\def\ylim{1.5cm}
\def\KHight{1.5cm}
\def\KWidth{5.5cm}
\node [rectangle, draw, minimum width=\plantWidth, minimum height=\boxHight, align=center] (plant) at (0,0) {Plant}; 
\node [rectangle, draw, minimum width=\encWidth, minimum height=\boxHight, align=center] (encoder) at (\encCenter,-\ylim) {$\mathbf{E}$}; 
\node [rectangle, draw, minimum width=\encWidth, minimum height=\boxHight, align=center] (decoder) at (-\encCenter,-\ylim) {$\mathbf{D}$}; 
\node [dashed, rectangle, draw, minimum width=\KWidth, minimum height=\KHight] (K) at (0,-\ylim) {}; 

\draw [dotted, ->,line width=0.7] (encoder) -- node[above]{$\mathbf{m}$} (decoder); 
\draw [->,,line width=0.5] (plant) -- node[above left]{$\mathbf{y}$} (\xlim,0) -- (\xlim,-\ylim) -- (encoder); 
\draw [->,,line width=0.5] (decoder) -- (-\xlim,-\ylim) -- (-\xlim,0) -- node[above right]{$\mathbf{u}$} (plant);

\draw (\encCenter+0.5cm,-\ylim+\KHight/2+0.25cm) node{$\mathbf{K}$};
\end{tikzpicture}

%% file: root_arxiv.bbl
\begin{thebibliography}{10}
\providecommand{\url}[1]{#1}
\csname url@samestyle\endcsname
\providecommand{\newblock}{\relax}
\providecommand{\bibinfo}[2]{#2}
\providecommand{\BIBentrySTDinterwordspacing}{\spaceskip=0pt\relax}
\providecommand{\BIBentryALTinterwordstretchfactor}{4}
\providecommand{\BIBentryALTinterwordspacing}{\spaceskip=\fontdimen2\font plus
\BIBentryALTinterwordstretchfactor\fontdimen3\font minus
  \fontdimen4\font\relax}
\providecommand{\BIBforeignlanguage}[2]{{%
\expandafter\ifx\csname l@#1\endcsname\relax
\typeout{** WARNING: IEEEtran.bst: No hyphenation pattern has been}%
\typeout{** loaded for the language `#1'. Using the pattern for}%
\typeout{** the default language instead.}%
\else
\language=\csname l@#1\endcsname
\fi
#2}}
\providecommand{\BIBdecl}{\relax}
\BIBdecl

\bibitem{cervin2003does}
A.~Cervin, D.~Henriksson, B.~Lincoln, J.~Eker, and K.-E. Arzen, ``How does
  control timing affect performance? {A}nalysis and simulation of timing using
  jitterbug and truetime,'' \emph{IEEE Contr. Systems Magazine}, vol.~23,
  no.~3, pp. 16--30, 2003.

\bibitem{aspeel2021optimal}
A.~Aspeel, K.~Rutledge, R.~M. Jungers, B.~Macq, and N.~{\"O}zay, ``Optimal
  control for linear networked control systems with information transmission
  constraints,'' in \emph{2021 60th IEEE Conf. on Decision and Contr.
  (CDC)}.\hskip 1em plus 0.5em minus 0.4em\relax IEEE, 2021, pp. 1960--1967.

\bibitem{balaghi20192}
M.~H. Balaghi, D.~J. Antunes, and W.~M. Heemels, ``An l 2-consistent data
  transmission sequence for linear systems,'' in \emph{2019 IEEE 58th Conf. on
  Decision and Contr. (CDC)}.\hskip 1em plus 0.5em minus 0.4em\relax IEEE,
  2019, pp. 2622--2627.

\bibitem{long2017stochastic}
Y.~Long, S.~Liu, and L.~Xie, ``Stochastic channel allocation for networked
  control systems,'' \emph{IEEE Contr. Systems Letters}, vol.~1, no.~1, pp.
  176--181, 2017.

\bibitem{heemels2012introduction}
W.~P. Heemels, K.~H. Johansson, and P.~Tabuada, ``An introduction to
  event-triggered and self-triggered control,'' in \emph{2012 IEEE 51st IEEE
  Conf. on Decision and Contr. (CDC)}.\hskip 1em plus 0.5em minus 0.4em\relax
  IEEE, 2012, pp. 3270--3285.

\bibitem{BRAKSMAYER20172633}
M.~Braksmayer and L.~Mirkin, ``Redesign of stabilizing discrete-time
  controllers to accommodate intermittent sampling,'' \emph{IFAC-PapersOnLine},
  vol.~50, no.~1, pp. 2633--2638, 2017.

\bibitem{cho2023lowrank}
M.~Cho, A.~Abdallah, and M.~Rasouli, ``Low-rank {LQR} optimal control design
  for controlling distributed multi-agent systems,'' in \emph{2023 European
  Contr. Conf. (ECC)}, 2023, pp. 1--6.

\bibitem{antonelli2013interconnected}
G.~Antonelli, ``Interconnected dynamic systems: An overview on distributed
  control,'' \emph{IEEE Contr. Systems Magazine}, vol.~33, no.~1, pp. 76--88,
  2013.

\bibitem{wang2018convex}
Y.~Wang, J.~A. Lopez, and M.~Sznaier, ``Convex optimization approaches to
  information structured decentralized control,'' \emph{IEEE Trans. on Autom.
  Contr.}, vol.~63, no.~10, pp. 3393--3403, 2018.

\bibitem{skaf2010design}
J.~Skaf and S.~P. Boyd, ``Design of affine controllers via convex
  optimization,'' \emph{IEEE Trans. on Autom. Contr.}, vol.~55, no.~11, pp.
  2476--2487, 2010.

\bibitem{anderson2019system}
J.~Anderson, J.~C. Doyle, S.~H. Low, and N.~Matni, ``System level synthesis,''
  \emph{Annual Reviews in Contr.}, vol.~47, pp. 364--393, 2019.

\bibitem{rotkowitz2005characterization}
M.~Rotkowitz and S.~Lall, ``A characterization of convex problems in
  decentralized control,'' \emph{IEEE Trans. on Autom. Contr.}, vol.~50,
  no.~12, pp. 1984--1996, 2005.

\bibitem{lessard2011quadratic}
L.~Lessard and S.~Lall, ``Quadratic invariance is necessary and sufficient for
  convexity,'' in \emph{Proc. of the 2011 American Contr. Conf.}\hskip 1em plus
  0.5em minus 0.4em\relax IEEE, 2011, pp. 5360--5362.

\bibitem{alonso2022distributed}
C.~Amo~Alonso, J.~S. Li, J.~Anderson, and N.~Matni, ``Distributed and localized
  model predictive control. part i: Synthesis and implementation,'' \emph{IEEE
  Trans. on Contr. of Network Systems}, 2022.

\bibitem{lin2013design}
F.~Lin, M.~Fardad, and M.~R. Jovanovi{\'c}, ``Design of optimal sparse feedback
  gains via the alternating direction method of multipliers,'' \emph{IEEE
  Trans. on Autom. Contr.}, vol.~58, no.~9, pp. 2426--2431, 2013.

\bibitem{matni2016regularization}
N.~Matni and V.~Chandrasekaran, ``Regularization for design,'' \emph{IEEE
  Trans. on Autom. Contr.}, vol.~61, no.~12, pp. 3991--4006, 2016.

\bibitem{siami2020deterministic}
M.~Siami, A.~Olshevsky, and A.~Jadbabaie, ``Deterministic and randomized
  actuator scheduling with guaranteed performance bounds,'' \emph{IEEE Trans.
  on Autom. Contr.}, vol.~66, no.~4, pp. 1686--1701, 2020.

\bibitem{tatikonda2004control}
S.~Tatikonda and S.~Mitter, ``Control under communication constraints,''
  \emph{IEEE Trans. on Autom. Contr.}, vol.~49, no.~7, pp. 1056--1068, 2004.

\bibitem{nair2007feedback}
G.~N. Nair, F.~Fagnani, S.~Zampieri, and R.~J. Evans, ``Feedback control under
  data rate constraints: An overview,'' \emph{Proc. of the IEEE}, vol.~95,
  no.~1, pp. 108--137, 2007.

\bibitem{chen2019system}
Y.~Chen and J.~Anderson, ``System level synthesis with state and input
  constraints,'' in \emph{2019 IEEE 58th Conf. on Decision and Contr.
  (CDC)}.\hskip 1em plus 0.5em minus 0.4em\relax IEEE, 2019, pp. 5258--5263.

\bibitem{hassaan2022system}
S.~M. Hassaan and S.~Z. Yong, ``System-level recurrent state estimators for
  affine systems subject to data losses modeled by automata,'' in \emph{2022
  IEEE 61st Conf. on Decision and Contr. (CDC)}.\hskip 1em plus 0.5em minus
  0.4em\relax IEEE, 2022, pp. 4118--4124.

\bibitem{mohan2010reweighted}
K.~Mohan and M.~Fazel, ``Reweighted nuclear norm minimization with application
  to system identification,'' in \emph{Proc. of the 2010 American Contr.
  Conf.}\hskip 1em plus 0.5em minus 0.4em\relax IEEE, 2010, pp. 2953--2959.

\bibitem{candes2008enhancing}
E.~J. Candes, M.~B. Wakin, and S.~P. Boyd, ``Enhancing sparsity by reweighted
  l1 minimization,'' \emph{J. of Fourier analysis and applications}, vol.~14,
  pp. 877--905, 2008.

\end{thebibliography}
